\documentclass[
aps,
reprint,
superscriptaddress,
nofootinbib,onecolumn
]{revtex4-2}%
\usepackage{amsmath}
\usepackage{amsfonts}
\usepackage{amsthm}
\usepackage{amssymb}
\usepackage{graphicx}%
\usepackage[svgnames]{xcolor}
\usepackage[colorlinks=true,citecolor=blue]{hyperref}
\usepackage{url}

\newtheorem{proposition}{Proposition}
\usepackage{empheq}
\providecommand{\U}[1]{\protect\rule{.1in}{.1in}}
\begin{document}
\title{Extremal curves of Perlick's proper time in Weyl geometry}

\author{E. Rodrigues}
\email{dota.jp13@gmail.com}
\affiliation{
Departamento de F\'{\i}sica,
Universidade Federal da Para\'{\i}ba,
Caixa Postal 5008,
58059-970 Jo\~ao Pessoa, PB, Brazil
}

\author{F. Dahia}
\email{fdahia@fisica.ufpb.br}
\affiliation{
Departamento de F\'{\i}sica,
Universidade Federal da Para\'{\i}ba,
Caixa Postal 5008,
58059-970 Jo\~ao Pessoa, PB, Brazil
}

\author{I. P. Lobo}
\email{iarley\_lobo@fisica.ufpb.br}
\affiliation{
Departamento de F\'{\i}sica,
Universidade Federal da Para\'{\i}ba,
Caixa Postal 5008,
58059-970 Jo\~ao Pessoa, PB, Brazil
}
\affiliation{
Physics Department,
Federal University of Campina Grande,
Caixa Postal 10071,
58429-900 Campina Grande, Para\'{\i}ba, Brazil
}

\author{C. Romero}
\email{cromero@fisica.ufpb.br}
\affiliation{
Departamento de F\'{\i}sica,
Universidade Federal da Para\'{\i}ba,
Caixa Postal 5008,
58059-970 Jo\~ao Pessoa, PB, Brazil
}

\begin{abstract}
We investigate the extremization of Perlick's proper time in non-integrable Weyl geometry. We derive the corresponding generalized Euler--Lagrange equations and show that the resulting extremals do not, in general, coincide with the autoparallels of the Weyl connection, their difference being governed by the Weyl length curvature and vanishing in the integrable case. A distinctive feature of the extremal equation is its nonlocal character: in the Weyl proper-time parametrization, the acceleration depends explicitly on the remaining proper time to the endpoint of the variational interval. We illustrate this behavior for a weak constant Weyl field and show that a local Lorentz-force dynamics emerges in a double-scaling limit in which the Weyl field vanishes and the terminal proper time diverges while their product remains finite. These results uncover a nontrivial relation between proper-time extremization, Weyl non-integrability, and local force dynamics.
\end{abstract}

%\pacs{04.20.Jb, 11.10.kk, 98.80.Cq}
\maketitle

Keywords: Weyl geometry; proper time; variational principles; nonlocal dynamics.

\section{Introduction}

Weyl geometry provides one of the simplest extensions of pseudo-Riemannian geometry
in which the comparison of lengths at different spacetime points becomes path
dependent. In addition to the metric tensor $g_{\mu\nu}$, the geometry is
characterized by a $1$-form $\sigma_{\mu}$, which governs the change of
lengths under parallel transport. Its exterior derivative $d\sigma$, which in
local coordinates has the components $F_{\mu\nu}=\partial_{\mu}\sigma_{\nu
}-\partial_{\nu}\sigma_{\mu}$, measures the non-integrability of the length
transport and vanishes in the integrable Weyl case. Historically, this
structure was introduced by Weyl in an attempt to provide a geometrical
description of gravitation and electromagnetism \cite{Weyl1918}, and it later
acquired a broader significance in some geometric approaches to gravitation
and spacetime structure.

Weyl geometry also arises naturally in the axiomatic construction of spacetime
proposed by Ehlers, Pirani and Schild \cite{EPS1972}. In the EPS framework,
light propagation determines a conformal structure, while the trajectories of
freely falling particles determine a projective structure. Their compatibility
leads, in general, to a Weyl geometry rather than directly to a
pseudo-Riemannian one. In this setting, the operational definition of proper
time becomes particularly important, since the metric norm of the tangent
vector is not, by itself, invariant under Weyl transformations. Related to the
EPS approach, we would like to refer the reader to some results in the
literature that explore into the Weylian structures
\cite{Perlick1987,Koehler1978,Perlick1991,CastagninoHarari1983,Audretsch:1984hd,Audretsch:1991aj}%
.

In particular, a Weyl-invariant notion of proper time was proposed by Perlick
\cite{Perlick1987}. His construction provides a characterization of standard
clocks that is compatible with the Weyl gauge structure and has subsequently
been discussed in connection with the second clock effect and its possible
phenomenological consequences
\cite{Avalos2018,LoboRomero2018,Delhom:2020vpe,Khodadi:2026zoi}. Unlike the
usual pseudo-Riemannian arc length, however, Perlick's proper-time functional
contains an accumulated contribution of the Weyl $1$-form along the worldline.
This raises the natural  question: What curves extremize this proper time, and
how are they related to the autoparallels of the Weyl connection?

In pseudo-Riemannian geometry, the extremization of proper time leads, up to
parametrization, to the geodesic equation. One might therefore expect an
analogous correspondence in Weyl geometry. We show that this expectation fails
in the non-integrable case, i.e. when $d\sigma$ is not a closed form. If
$\sigma$ is an exact form, then we say that we have a \textit{Weyl integrable
spacetime (WIST) }\cite{Novello1992}. The extremals of Perlick's proper time
satisfy a generalized Euler--Lagrange equation containing an additional
contribution proportional to the so-called \textit{Weyl length curvature}.
(Let us recall that this kind of curvature was introduced by Weyl and was
identified to the Faraday tensor of Electromagnetism  \cite{Weyl1918}.) As a
consequence, they do not coincide, in general, with the autoparallels of the
Weyl connection, as we will show in this paper. The two notions of motion agree only when the relevant
curvature contribution vanishes. The quest for the existence of a variational
principle for Weyl geodesics has also been recently explored in
\cite{Heisenberg:2026ess}.

A distinctive feature of the equation which gives the extremal of Perlick's
proper time is its nonlocal structure. The acceleration at a given point
depends on an integral over the remaining segment of the worldline up to the
final endpoint of the variational interval. In the Weyl proper-time
parametrization, this dependence simplifies and becomes explicitly
proportional to the remaining proper-time interval. The resulting dynamics
therefore has a \textit{terminal }character: the
local equation retains information about the future boundary entering the
original variational problem. This feature should be understood primarily as a
property of the boundary-value formulation, although it becomes conceptually
relevant if the resulting extremal equation is interpreted as a fundamental
equation of motion.

We analyze this structure explicitly in a simple $1+1$ dimensional model with
constant Weyl length curvature. In the weak-field regime, the trajectories can
be obtained analytically and compared with both Minkowski geodesics and
ordinary Weyl autoparallels. This comparison makes it possible to separate the
effect of the Weyl connection from the additional contribution generated by
the proper-time extremization. The full nonlinear equations are also solved
numerically, confirming the perturbative behavior.

The dependence on the final endpoint leads to a second interesting limit.
Simply increasing the terminal proper time at fixed Weyl field enhances,
rather than suppresses, the nonlocal contribution. However, a nontrivial local
limit is obtained by simultaneously weakening the Weyl field and moving the
terminal point to infinity while keeping their product fixed. In this
double-scaling regime, the direct corrections associated with the Weyl
connection disappear, whereas the curvature-dependent contribution remains
finite. The resulting equation coincides with the relativistic Lorentz-force
equation. We illustrate this mechanism numerically and show explicitly the
convergence of the Perlick trajectories towards Lorentz-force motion.

The paper is organized as follows. We first review the relevant elements of
Weyl geometry and Perlick's definition of proper time. We then formulate the
generalized variational problem and derive the corresponding Euler--Lagrange
equations. Their specialization to Weyl geometry leads to a nonlocal equation
governing the extremals of Perlick's proper time, whose terminal structure is
subsequently analyzed. We then consider a weak constant Weyl field and compare
the resulting trajectories with Minkowski geodesics and Weyl autoparallels.
Finally, we investigate the large terminal-time behavior and show how the
Lorentz force emerges in the appropriate double-scaling limit.

%%%%%%%%%%%%%%%%%%%%%%%%%%%%%%%%%%%%%%%%%%%%%%%%%%%%%%%%%%%%%%%%%%%%%%%%%%%%%%%%%%%%%%%%%%%%%%%%%%%%%%%%%%%%%%%%

\section{Perlick's proper time}

\label{sec:perlick}

In a Weyl spacetime, the usual pseudo-Riemannian prescription for measuring
proper time requires some care.\footnote{We adopt the notation of
\cite{Spivak1999,ONeill1983}.} The geometry is characterized by a metric $g$, a
$1$-form $\sigma$, and a torsion-free affine connection satisfying
\begin{equation}
\nabla_{\lambda}g_{\mu\nu}=\sigma_{\lambda}g_{\mu\nu}%
.\label{eq:weylcompatibility}%
\end{equation}
Under a Weyl transformation, the metric and the Weyl $1$-form transform as
$g\rightarrow e^{\phi}g$ and $\sigma\rightarrow\sigma+d\phi$, while the affine
connection remains unchanged. Thus, although a particular metric
representative can be used to perform calculations, physical prescriptions
should ultimately be independent of this choice. This becomes particularly
relevant when defining the time measured by an observer. In pseudo-Riemannian
geometry, proper time is directly associated with the metric length of a
timelike curve. In a general Weyl geometry, however, the norm of a vector
changes under parallel transport according to the Weyl $1$-form. Consequently,
the usual metric arc length cannot be adopted without taking into account the
additional geometric structure. Perlick addressed this problem by giving an
operational characterization of proper time in terms of standard clocks
\cite{Perlick1987}. Here we briefly revisit his procedure in the form
developed in Ref.~\cite{Avalos2018}, which will lead to the functional to be
extremized in the following sections.

A timelike curve $\gamma$ parametrized by $\tau$ is called a standard clock if
its acceleration is orthogonal to its tangent vector, that is,
\begin{equation}
g\left(  \dot{\gamma},\frac{D\dot{\gamma}}{d\tau}\right)
=0.\label{eq:standardclock}%
\end{equation}
Proper time is then identified with a parametrization for which the worldline
satisfies Eq.~\eqref{eq:standardclock}. This condition is particularly
convenient in Weyl geometry because orthogonality is preserved under conformal
transformations, while the affine connection, and hence the covariant
acceleration, is invariant under Weyl transformations. The standard-clock
condition therefore does not depend on the representative $(g,\sigma)$ chosen
within the Weyl conformal struture.

Consider now a timelike curve $\gamma(t)$ written in terms of an arbitrary
parameter $t$, and suppose that $\tau=\tau(t)$ is a reparametrization that
turns it into a standard clock. Applying Eq.~\eqref{eq:standardclock} to the
reparametrized curve gives \footnote{In this paper our convention for the
metric signature is $(---+)$. }
\begin{equation}
\frac{d^{2}\tau}{dt^{2}}-\frac{g\left(  \dot{\gamma},D\dot{\gamma}/dt\right)
}{g(\dot{\gamma},\dot{\gamma})}\frac{d\tau}{dt}=0.\label{eq:perlick-reparam}%
\end{equation}
Thus, the problem of defining proper time along an arbitrary timelike
worldline reduces to determining the reparametrization that solves Eq.~\eqref{eq:perlick-reparam}.

The Weyl compatibility condition allows this equation to be integrated
explicitly. Along the curve one has
\[
\frac{d}{dt}\ln\!\left[  -g(\dot{\gamma},\dot{\gamma})\right]  =2\frac
{g\left(  \dot{\gamma},D\dot{\gamma}/dt\right)  }{g(\dot{\gamma},\dot{\gamma
})}+\sigma_{\mu}\dot{x}^{\mu}.
\]
Using this relation in Eq.~\eqref{eq:perlick-reparam}, the equation for the
proper-time parametrization becomes
\[
\frac{d^{2}\tau}{dt^{2}}-\frac{1}{2}\left\{  \frac{d}{dt}\ln\!\left[
-g(\dot{\gamma},\dot{\gamma})\right]  -\sigma_{\mu}\dot{x}^{\mu}\right\}
\frac{d\tau}{dt}=0.
\]
Introducing $\psi=d\tau/dt$, this is reduced to a first-order linear equation
and can be immediately integrated. The result is
\begin{equation}
\frac{d\tau}{dt}=C\sqrt{-g_{\mu\nu}\dot{x}^{\mu}\dot{x}^{\nu}}\exp\left[
-\frac{1}{2}\int_{t_{0}}^{t}\sigma_{\mu}(x(s))\dot{x}^{\mu}(s)\,ds\right]
,\label{eq:perlick-differential}%
\end{equation}
where $C$ is determined by the normalization of the clock at the reference
point $t_{0}$.

A second integration gives the elapsed proper time between $t_{0}$ and $t$,
\begin{equation}
\Delta\tau(t)=C\int_{t_{0}}^{t}\exp\left[  -\frac{1}{2}\int_{t_{0}}^{u}%
\sigma_{\mu}(x(s))\dot{x}^{\mu}(s)\,ds\right]  \sqrt{-g_{\mu\nu}(x(u))\dot
{x}^{\mu}(u)\dot{x}^{\nu}(u)}\,du.\label{eq:perlick-proper-time}%
\end{equation}
The constant $C$ fixes the overall scale of the clock and will play no role in
the extremization considered below and is given by $\left.  d\tau
/dt\right\vert _{t=t_{0}}\left[  -g\bigl(\dot{\gamma}(t_{0}),\dot{\gamma
}(t_{0})\bigr)\right]  ^{-1/2}$.

Equation~\eqref{eq:perlick-proper-time} differs from the usual pseudo-Riemannian
proper time by the exponential factor involving the integral of the Weyl
$1$-form along the worldline. This factor compensates for the change of the
metric norm under Weyl transformations and incorporates the history of Weyl
transport between the reference point $x(t_{0})$ and the point $x(u)$. The
expression is consequently invariant under both Weyl transformations and
reparametrizations of the curve. When $\sigma=0$, the exponential factor
becomes unity and Eq.~\eqref{eq:perlick-proper-time} reduces to the familiar
relativistic proper-time functional.

The mathematical and operational properties of this definition, including its
relation to standard clocks, additivity, and the second clock effect, have
been discussed in detail in Ref.~\cite{Avalos2018}. We shall not revisit those
issues here. Instead, our starting point will be
Eq.~\eqref{eq:perlick-proper-time} itself. Once Perlick's proper time is
regarded as a functional of the worldline, a natural question arises: which
curves make this functional stationary? In ordinary pseudo-Riemannian geometry
this procedure leads to the geodesic equation. In the following sections, we
investigate the corresponding variational problem in a general Weyl geometry
and compare its extremals with the autoparallels of the Weyl connection.

%%%%%%%%%%%%%%%%%%%%%%%%%%%%%%%%%%%%%%%%%%%%%%%%%%%%%%%%%%%%%%%%%%%%%%%%%%%%%%%%%%%%%%%%%%%%%%%%%%%%%%%%%%%%%%%%

\section{Variation of the general Perlick functional}

\label{sec:var}

At this point, we shall work directly with a coordinate system, so we will
refer to $x^{\alpha}=x^{\alpha}(s)$ as a curve on $M$, parametrized by
$s\in\mathbb{R}$, and a dot denotes differentiation with respect to this
parameter. The proper time \eqref{eq:perlick-proper-time} defines a
functional, which we shall vary:
\begin{equation}
S[x^{\alpha}(s)]=\int_{a}^{b}\exp\left(  \int_{a}^{s}l\bigl(x^{\alpha}%
(t),\dot{x}^{\alpha}(t)\bigr)\,dt\right)  L\bigl(x^{\alpha}(s),\dot{x}%
^{\alpha}(s)\bigr)\,ds.\label{eq:new_perlick}%
\end{equation}
The function
\begin{equation}
L(x,\dot{x})=\frac{\left.  d\tau/dt\right\vert _{t=t_{0}}}{\left[
-g\bigl(\dot{\gamma}(t_{0}),\dot{\gamma}(t_{0})\bigr)\right]  ^{1/2}}\left[
-g\bigl(\dot{\gamma}(s),\dot{\gamma}(s)\bigr)\right]  ^{1/2}%
\end{equation}
is the local part of the Lagrangian, whereas
\begin{equation}
l(x,\dot{x})=-\frac{1}{2}\sigma_{\alpha}(x)\dot{x}^{\alpha}=-\frac{1}{2}%
\sigma\bigl(\dot{\gamma}(s)\bigr)
\end{equation}
determines its nonlocal part.

We introduce a one-parameter family of curves $\eta^{\alpha}(s,\varepsilon)$
satisfying
\begin{equation}
\eta^{\alpha}(s,0)=x^{\alpha}(s), \qquad\eta^{\alpha}(a,\varepsilon
)=x^{\alpha}(a), \qquad\eta^{\alpha}(b,\varepsilon)=x^{\alpha}(b).
\end{equation}
The corresponding variation vector field is defined by
\begin{equation}
\xi^{\alpha}(s) \equiv\left.  \frac{\partial\eta^{\alpha}(s,\varepsilon)}
{\partial\varepsilon} \right|  _{\varepsilon=0}, \qquad\xi^{\alpha}%
(a)=\xi^{\alpha}(b)=0.
\end{equation}
For each $\varepsilon\in(-\epsilon,\epsilon)\subset\mathbb{R}$, the map
$s\mapsto\eta^{\alpha}(s,\varepsilon)$ represents a possible trajectory.

The functional evaluated along the varied curve is
\begin{equation}
S[\eta^{\alpha}(s,\varepsilon)] = \int_{a}^{b} \exp\left[  \int_{a}^{s}
l\bigl(\eta^{\alpha}(t,\varepsilon), \dot{\eta}^{\alpha}(t,\varepsilon
)\bigr)\,dt \right]  L\bigl(\eta^{\alpha}(s,\varepsilon), \dot{\eta}^{\alpha
}(s,\varepsilon)\bigr)\,ds.
\end{equation}
To extremize this functional, we calculate $\partial S/\partial\varepsilon$.
The details of this calculation are presented in Appendix \ref{app:details}. A
direct computation gives
\begin{align}
\frac{\partial S}{\partial\varepsilon} ={}  &  \int_{a}^{b} \exp\left(
\int_{a}^{s} l\bigl(\eta(t,\varepsilon),\dot{\eta}(t,\varepsilon)\bigr)\,dt
\right)  \Bigg\{ \Bigg[ \int_{a}^{s} \left(  \frac{\partial l}{\partial
\eta^{\alpha}} - \frac{d}{dt} \frac{\partial l}{\partial\dot{\eta}^{\alpha}}
\right)  \frac{\partial\eta^{\alpha}}{\partial\varepsilon}\,dt +
\frac{\partial l}{\partial\dot{\eta}^{\alpha}} \frac{\partial\eta^{\alpha}%
}{\partial\varepsilon} \Bigg]_{(s,\varepsilon)} L\nonumber\\
&  \qquad\qquad+ \left[  \frac{\partial L}{\partial\eta^{\alpha}} - \frac
{d}{ds} \left(  \frac{\partial L}{\partial\dot{\eta}^{\alpha}} \right)  -
l\frac{\partial L}{\partial\dot{\eta}^{\alpha}} \right]  _{(s,\varepsilon)}
\frac{\partial\eta^{\alpha}}{\partial\varepsilon} \Bigg\}\,ds + \int_{a}^{b}
\frac{d}{ds} \left[  \exp\left(  \int_{a}^{s} l\bigl(\eta,\dot{\eta}\bigr)\,dt
\right)  \frac{\partial L}{\partial\dot{\eta}^{\alpha}} \frac{\partial
\eta^{\alpha}}{\partial\varepsilon} \right]  ds .
\end{align}
Since the variation vanishes at the endpoints, the total derivative term does
not contribute. Evaluating the remaining terms at $\varepsilon=0$, we obtain
\begin{align}
\left.  \frac{\partial S}{\partial\varepsilon} \right|  _{\varepsilon=0} ={}
&  \int_{a}^{b} \exp\left[  \int_{a}^{s} l\bigl(x(t),\dot{x}(t)\bigr)\,dt
\right]  \Bigg\{ \left[  \int_{a}^{s} \left(  \frac{\partial l}{\partial
x^{\alpha}} - \frac{d}{dt} \frac{\partial l}{\partial\dot{x}^{\alpha}}
\right)  _{\!t} \xi^{\alpha}(t)\,dt \right]  L\nonumber\\
&  \hspace{2cm} + \left[  \frac{\partial L}{\partial x^{\alpha}} - \frac
{d}{ds} \left(  \frac{\partial L}{\partial\dot{x}^{\alpha}} \right)  +
\frac{\partial l}{\partial\dot{x}^{\alpha}}L - l\frac{\partial L}{\partial
\dot{x}^{\alpha}} \right]  _{\!s} \xi^{\alpha}(s) \Bigg\}\,ds.
\label{eq:first-variation-components}%
\end{align}
In Appendix \ref{app:general-var}, we establish the result for an arbitrary
variation. For the moment, consider a variation localized at $s=\lambda$ in a
fixed coordinate direction, namely
\begin{equation}
\xi^{\alpha}(s) = v^{\alpha}\delta(s-\lambda),
\end{equation}
where $v^{\alpha}$ is an arbitrary constant vector. The equations of motion
then take the form%

\begin{equation}
\frac{\partial L}{\partial x^{\alpha}}-\frac{d}{dt}\left(  \frac{\partial
L}{\partial\dot{x}^{\alpha}}\right)  +\frac{\partial l}{\partial\dot
{x}^{\alpha}}L-l\frac{\partial L}{\partial\dot{x}^{\alpha}}=\left[  \frac
{d}{dt}\left(  \frac{\partial l}{\partial\dot{x}^{\alpha}}\right)
-\frac{\partial l}{\partial x^{\alpha}}\right]  \int_{t}^{b}\exp\left(
\int_{t}^{s}l(x,\dot{x})\,du\right)  L(x,\dot{x})\,ds.\label{eq:main}%
\end{equation}
This equation holds for each component $\alpha$. In Appendix \ref{app:details}%
, we prove that it constitutes a necessary and sufficient condition for the
extremization of the Perlick functional \eqref{eq:new_perlick}. It generalizes
the Euler--Lagrange equations for a nonlocal Lagrangian of the kind
\eqref{eq:new_perlick}. This equation can be applied for the Weyl case, which
will be the topic of the next section, but also to the generalization of
Perlick functional to the case of general non-metricity
\cite{Delhom:2020vpe,Delhom:2019yeo}.

%%%%%%%%%%%%%%%%%%%%%%%%%%%%%%%%%%%%%%%%%%%%%%%%%%%%%%%%%%%%%%%%%%%%%%%%%%%%%%%%

\section{Equation of motion in Weyl geometry}

We now specialize the generalized variational principle to Weyl geometry by
choosing $l(x,\dot{x})=-\frac{1}{2}\sigma_{\mu}(x)\dot{x}^{\mu}$, where
$\sigma_{\mu}$ is the Weyl $1$-form. The corresponding Euler--Lagrange
derivative is $\frac{d}{dt}\left(  \frac{\partial l}{\partial\dot{x}^{\alpha}%
}\right)  -\frac{\partial l}{\partial x^{\alpha}}=\frac{1}{2}F_{\alpha\beta
}\dot{x}^{\beta}$, where $F_{\mu\nu}=\partial_{\mu}\sigma_{\nu}-\partial_{\nu
}\sigma_{\mu}$ is the Weyl length-curvature $2$-form.

Introducing $\rho=\sqrt{-g_{\mu\nu}\dot{x}^{\mu}\dot{x}^{\nu}}$, a
straightforward computation shows that the generalized Euler--Lagrange
equations become
\begin{equation}
\frac{\widehat{D}\dot{x}^{\lambda}}{dt}=\left(  \frac{\dot{\rho}}{\rho}%
-\frac{1}{2}\sigma_{\mu}\dot{x}^{\mu}\right)  \dot{x}^{\lambda}+\frac{\rho}%
{2}\mathcal{J}(t)F^{\lambda}{}_{\beta}\dot{x}^{\beta},\label{eq:weylintegro}%
\end{equation}
where $\widehat{D}$ denotes the covariant derivative associated with the Weyl
connection
\[
\widehat{\Gamma}^{\lambda}{}_{\mu\nu}=\Gamma^{\lambda}{}_{\mu\nu}-\frac{1}%
{2}\left(  \delta_{\mu}^{\lambda}\sigma_{\nu}+\delta_{\nu}^{\lambda}%
\sigma_{\mu}-g_{\mu\nu}\sigma^{\lambda}\right)  ,
\]
and
\begin{equation}
\mathcal{J}(t)=\int_{t}^{b}\exp\left[  -\frac{1}{2}\int_{t}^{s}\sigma_{\mu
}(x(u))\dot{x}^{\mu}(u)\,du\right]  \rho(s)\,ds.\label{eq:J}%
\end{equation}

Equation~\eqref{eq:weylintegro} is a nonlinear integro-differential equation.
Its differential part is of second order in the trajectory, whereas its
right-hand side contains the nonlocal functional~\eqref{eq:J}. At a given
value of the parameter $t$, the quantity $\mathcal{J}(t)$ depends on the
values of the trajectory over the entire interval $[t,b]$. Consequently, the
acceleration at $t$ cannot, in general, be determined solely from the
instantaneous values of $x^{\mu}(t)$ and $\dot{x}^{\mu}(t)$, but also depends
on the remaining segment of the worldline up to the final endpoint.

The exponential factor $\exp\!\left[  -\frac{1}{2}\int_{t}^{s}\sigma_{\mu
}dx^{\mu}\right]  $ has a natural interpretation in Weyl geometry: it
represents the scale factor associated with Weyl transport between the points
$x(t)$ and $x(s)$. Accordingly, $\mathcal{J}(t)$ may be viewed as the
remaining metric length of the worldline weighted by the accumulated Weyl
rescaling along the curve. The nonlocal contribution therefore couples the
local Weyl length curvature $F_{\mu\nu}$ to a global property of the trajectory.

The nonlocal behavior disappears in the Weyl integrable case
\cite{Novello1992,SalimSautu1996,OliveiraSalimSautu1997,Melnikov1995,Bronnikov1995,Miritzis2004,AguilarRomero2009a,AguilarRomero2009b,Miritzis2013,PoulisSalim2014,Vazirian2015,Pucheu2016,Scholz:2017pfo,Poulis:2013rma,Almeida:2013dba}%
, in which $\sigma=d\phi$, where $\phi$ is a scalar field. In this case, the
Perlick time is given by an invariant prescription described in
\cite{Romero:2012hs}, and the trajectories are Weyl geodesics.

We may simplify the equation considerably by choosing the Weyl-invariant
proper time as the parameter along the worldline. Let $\tau$ denote this
parameter and define $u^{\mu}=dx^{\mu}/d\tau$ and $\rho=\sqrt{-g_{\mu\nu
}u^{\mu}u^{\nu}}$. By definition, the Weyl proper time satisfies
\[
\rho(\tau)\exp\left[  -\frac{1}{2}\int_{\tau_{0}}^{\tau}\sigma_{\mu}u^{\mu
}\,d\sigma\right]  =\mathrm{const}.
\]
The normalization of $\tau$ allows this constant to be chosen equal to unity.
Taking the logarithmic derivative then gives $\rho^{-1}d\rho/d\tau=\sigma
_{\mu}u^{\mu}/2$. Therefore, the term parallel to the tangent vector in
Eq.~\eqref{eq:weylintegro} vanishes, and the equation of motion reduces to
\begin{equation}
\frac{\widehat{D}u^{\lambda}}{d\tau}=\frac{\rho}{2}\mathcal{J}(\tau
)F^{\lambda}{}_{\beta}u^{\beta}.\label{eq:weyl-proper-J}%
\end{equation}

The nonlocal functional can also be evaluated explicitly in this
parametrization. Replacing $t$ by $\tau$ in Eq.~\eqref{eq:J}, we have
\begin{equation}
\mathcal{J}(\tau) = \rho(\tau) \left(  \tau_{b}-\tau\right)  .
\label{eq:J-proper-solution}%
\end{equation}

Substituting Eq.~\eqref{eq:J-proper-solution} into
Eq.~\eqref{eq:weyl-proper-J}, we finally obtain
\begin{equation}
 \frac{\widehat D u^\lambda}{d\tau} = \frac{\rho^2}{2} \left( \tau_b-\tau \right) F^\lambda{}_\beta u^\beta.
\label{eq:weyl-proper-nonlocal}%
\end{equation}

Thus, the use of Weyl proper time eliminates the tangential contribution and
allows the nonlocal integral to be carried out explicitly. Nevertheless, the
dependence on the final endpoint is not removed: it survives through the
factor $\tau_{b}-\tau$. In particular, the curvature-dependent contribution
vanishes at $\tau=\tau_{b}$. This makes explicit that the terminal character
of the original variational problem is not an artifact of an arbitrary
parametrization, but is inherited from the boundary structure of the nonlocal functional.

The dependence on the interval $[t,b]$ or $[\tau,\tau_{b}]$ gives the equation
an advanced, or terminal-value, character. In contrast with an ordinary local
equation of motion, specifying only the initial position and velocity is not
sufficient to determine the subsequent evolution. Instead, the complete
curve---or, equivalently, suitable terminal information---is required. This
does not make the variational problem mathematically inconsistent, since it
may still be interpreted as a boundary-value problem for extremal curves
joining two fixed endpoints. Nevertheless, it prevents a direct interpretation
of Eq.~\eqref{eq:weylintegro} as a conventional causal evolution equation.

The role of the final endpoint is reflected in the boundary condition
$\mathcal{J}(b)=0$, implying that the curvature-dependent force vanishes at
the endpoint selected in the variational problem. Since the point $b$ is not
distinguished by the local geometry, this behavior is determined entirely by
the boundary structure of the functional. Moreover, changing the endpoint
generally changes $\mathcal{J}(t)$ and therefore modifies the equation of
motion along the entire preceding segment of the worldline.

It is important to emphasize that this nonlocality cannot be removed by a
reparametrization of the curve. Indeed, the functional~\eqref{eq:J} is
invariant under regular orientation-preserving changes of the worldline
parameter, so that $\mathcal{J}(t)$ is associated with the geometric curve
itself rather than with its parametrization.

The integro-differential equation can nevertheless be rewritten as a local
differential system. Differentiating Eq.~\eqref{eq:J} with respect to $t$ and using Leibniz's rule, we obtain
\begin{align}
\dot{\mathcal J}(t)
&=
-\rho(t)
+
\int_t^b
\frac{\partial}{\partial t}
\exp\left[
-\frac12\int_t^s
\omega_\mu(x(u))\dot{x}^\mu(u)\,du
\right]
\rho(s)\,ds \nonumber\\
&=
-\rho(t)
+
\frac12\omega_\mu(x(t))\dot{x}^\mu(t)\mathcal J(t),
\end{align}
where we used
\[
\frac{\partial}{\partial t}
\int_t^s
\omega_\mu(x(u))\dot{x}^\mu(u)\,du
=
-\omega_\mu(x(t))\dot{x}^\mu(t).
\]
Hence,
\begin{equation}
\dot{\mathcal J}
=
\frac12\omega_\mu\dot{x}^\mu\mathcal J-\rho,
\qquad
\mathcal J(b)=0,
\end{equation}
with the terminal condition following directly from Eq.~\eqref{eq:J}.

Therefore, the integral may be replaced by an auxiliary scalar satisfying a
first-order differential equation. The price to pay is that the nonlocality is
transferred entirely to the terminal condition. Consequently, the system
remains exactly equivalent to the original variational problem, although the
evolution equations themselves become local.

%%%%%%%%%%%%%%%%%%%%%%%%%%%%%%%%%%%%%%%%%%%%%%%%%%%%%%%%%%%%%%%%%%%%%%%%%%%%%%%%%%%%%%

\subsection{Lorentz force limit of the nonlocal theory}

\label{subsec:lorentz-limit}

An interesting limit of the original nonlocal dynamics is obtained when the
Weyl $1$-form is weak while the final endpoint of the variational interval is
taken sufficiently distant. In this regime, the curvature-dependent
contribution can remain finite even though the Weyl field itself becomes
arbitrarily small. As we now show, an appropriate double-scaling limit
reproduces the ordinary Lorentz-force equation.

We assume, as in the original Weyl unified theory \cite{Weyl1918}, that the Weyl
$1$-form is proportional to an electromagnetic potential, $\sigma_{\mu
}=\lambda A_{\mu}$, where $\lambda$ parametrizes the strength of the Weyl
coupling. The corresponding length curvature is then $F_{\mu\nu}=\lambda
F_{\mu\nu}^{\mathrm{em}}$, with $F_{\mu\nu}^{\mathrm{em}}=\partial_{\mu}%
A_{\nu}-\partial_{\nu}A_{\mu}$.

In the Weyl proper-time parametrization, the equation of motion of the
original nonlocal theory is given by Eq.~\eqref{eq:weyl-proper-nonlocal}. With
the above identification, it becomes
\begin{equation}
\frac{\widehat D u^{\lambda}}{d\tau} = \frac{\lambda\rho^{2}}{2} (\tau
_{b}-\tau) F^{\mathrm{em}\,\lambda}{}_{\beta}u^{\beta}.
\label{eq:lorentz-limit-start}%
\end{equation}
This expression suggests considering the simultaneous limit
\begin{equation}
\lambda\rightarrow0, \qquad\tau_{b}\rightarrow\infty, \qquad\frac{\lambda
\tau_{b}}{2} \equiv\frac{q}{m} = \text{constant}. \label{eq:double-scaling}%
\end{equation}
For any finite proper-time interval, or more generally whenever $\tau/\tau
_{b}\ll1$, the coefficient appearing in Eq.~\eqref{eq:lorentz-limit-start} can
be written as
\[
\frac{\lambda}{2}(\tau_{b}-\tau) = \frac{q}{m} \left(  1-\frac{\tau}{\tau_{b}}
\right)  ,
\]
and therefore approaches $q/m$ in the limit \eqref{eq:double-scaling}.

At the same time, the remaining Weyl corrections disappear. Since $\sigma
_{\mu}=\mathcal{O}(\lambda)$, the Weyl connection approaches the Levi--Civita
connection as $\widehat{\Gamma}^{\lambda}{}_{\mu\nu}=\Gamma^{\lambda}{}%
_{\mu\nu}+\mathcal{O}(\lambda)$. Moreover, the Weyl proper-time condition
implies that $\rho=1+\mathcal{O}(\lambda)$ over any finite portion of the
worldline, assuming the usual normalization at the initial point.
Consequently, Eq.~\eqref{eq:lorentz-limit-start} reduces to
\begin{equation}
 \frac{D u^\lambda}{d\tau} = \frac{q}{m} F^{{\rm em}\,\lambda}{}_\beta u^\beta, \label{eq:lorentz-force-recovered}%
\end{equation}
which is precisely the relativistic Lorentz-force equation.

The ordinary electromagnetic dynamics can therefore be recovered from the
original Perlick extremals through a weak-field, large-endpoint limit. The
important point is that the local Weyl coupling tends to zero while the
product $\lambda\tau_{b}$ remains finite. In this sense, the finite
interaction strength results from the compensation between a weak local Weyl
field and the large interval entering the nonlocal variational problem.

For a large but finite $\tau_{b}$, Eq.~\eqref{eq:lorentz-limit-start} can
instead be interpreted as a Lorentz-force equation with a slowly varying
effective charge-to-mass ratio,
\begin{equation}
\left(  \frac{q}{m}\right)  _{\mathrm{eff}} = \frac{q}{m} \left(  1-\frac
{\tau}{\tau_{b}} \right)  . \label{eq:effective-charge}%
\end{equation}
Thus, ordinary Lorentz-force dynamics is approximately recovered whenever
$\tau\ll\tau_{b}$, with the leading departure characterized by the ratio
$\tau/\tau_{b}$. Additional corrections associated with the Weyl connection
and the normalization factor $\rho$ are suppressed by $\lambda$, which in the
double-scaling limit is itself of order $1/\tau_{b}$.

This limit provides an intriguing connection with Weyl's original motivation
for introducing the $1$-form $\sigma_{\mu}$. The antisymmetric length
curvature already has the tensorial structure of an electromagnetic field
strength, while the nonlocal proper-time extremization supplies an effective
coupling. However, this interpretation must be treated with some care. In the
original nonlocal theory, $\tau_{b}$ is the endpoint of the variational
interval rather than an intrinsic property of the particle. A direct
identification of $\lambda\tau_{b}/2$ with the physical ratio $q/m$ would
therefore require a physical prescription that fixes this combination
independently of an arbitrary choice of the endpoint. The strict
double-scaling limit avoids an explicit dependence on a finite endpoint within
any finite observation interval, but the physical origin of the fixed scale
$q/m$ would still need to be understood.

Nevertheless, Eq.~\eqref{eq:lorentz-force-recovered} shows that the Lorentz
force is contained as a well-defined scaling limit of the original nonlocal
Perlick dynamics. This provides an additional connection between the
variational structure studied here and Weyl's original geometrical
interpretation of electromagnetism.

%%%%%%%%%%%%%%%%%%%%%%%%%%%%%%%%%%%%%%%%%%%%%%%%%%%%%%%%%%%%%%%%%%%%%%%%%%%%%%%%%%%%%%%%%%%%%%%%%%%%%%%%%%%%%%%%%%%%%%%%%%%%%%%%%%%%%%%%%%%%%%%%%%%%%%%%%%%%%%%%%%%%%%%%%%%%%%%%%%%%%%%%%%

\section{Weak-field dynamics and the large terminal-time limit}

\label{sec:weak-field}

To illustrate the difference between Weyl autoparallels and the extremals of
Perlick's proper time, let us consider a simple constant-field configuration
in $1+1$ dimensions. We take a Minkowski representative, $g_{\mu\nu}=\eta
_{\mu\nu}=\mathrm{diag}(-1,1)$, and choose the Weyl $1$-form as $\sigma
_{0}=-Ex/2$ and $\sigma_{1}=Et/2$, so that the corresponding length curvature
is constant, $F_{01}=E$.

We consider a particle initially at rest, with $t(0)=x(0)=0$ and $u^{\mu
}(0)=(1,0)$, and use the Weyl proper time as parameter. In the absence of the
Weyl field, the trajectory reduces to the Minkowski geodesic
\begin{equation}
t_{\mathrm{M}}(\tau)=\tau, \qquad x_{\mathrm{M}}(\tau)=0.
\label{eq:minkowski-weak}%
\end{equation}

For a weak field, $E\ll1$, the spatial displacement is itself of order $E$.
Consequently, $\sigma_{\mu}u^{\mu}=\mathcal{O}(E^{2})$, while $\rho
=1+\mathcal{O}(E^{2})$ and $t(\tau)=\tau+\mathcal{O}(E^{2})$. The spatial
component of the Weyl autoparallel equation then becomes $du_{\mathrm{W}}%
^{1}/d\tau=E\tau/4+\mathcal{O}(E^{2})$. With the initial conditions above,
this gives
\begin{equation}
u_{\mathrm{W}}^{1}(\tau)=\frac{E}{8}\tau^{2},\qquad x_{\mathrm{W}}(\tau
)=\frac{E}{24}\tau^{3}+\mathcal{O}(E^{2}).\label{eq:weak-weyl-trajectory}%
\end{equation}
Thus, even though the representative metric is Minkowskian, the Weyl
connection produces a departure from inertial motion.

For the extremals of Perlick's proper time, the curvature-dependent
contribution must also be included. At the same order, the spatial equation
takes the form
\begin{equation}
\frac{du^{1}_{\mathrm{P}}}{d\tau} = \frac{E}{4}\tau- \frac{E}{2}(\tau_{b}%
-\tau) + \mathcal{O}(E^{2}) = E\left(  \frac{3}{4}\tau-\frac{\tau_{b}}{2}
\right)  + \mathcal{O}(E^{2}). \label{eq:weak-perlick-acceleration}%
\end{equation}
The first contribution originates from the Weyl connection, whereas the second
is associated with the extremization of Perlick's proper time and depends
explicitly on the remaining interval to the terminal point. Integration gives
\begin{equation}
u^{1}_{\mathrm{P}}(\tau) = E\left(  \frac{3}{8}\tau^{2}-\frac{\tau_{b}}{2}%
\tau\right)  , \qquad x_{\mathrm{P}}(\tau) = E\left(  \frac{\tau^{3}}{8} -
\frac{\tau_{b}\tau^{2}}{4} \right)  + \mathcal{O}(E^{2}).
\label{eq:weak-perlick-trajectory}%
\end{equation}

Equations~\eqref{eq:minkowski-weak}, \eqref{eq:weak-weyl-trajectory}, and
\eqref{eq:weak-perlick-trajectory} exhibit three distinct notions of motion.
The Minkowski trajectory remains at $x=0$, the Weyl autoparallel bends towards
positive $x$ for $E>0$, while the Perlick extremal initially bends in the
opposite direction, since $\ddot x_{\mathrm{P}}(0)=-E\tau_{b}/2$. The
difference between the latter two trajectories is
\[
x_{\mathrm{P}}(\tau)-x_{\mathrm{W}}(\tau) = E\left(  \frac{\tau^{3}}{12} -
\frac{\tau_{b}\tau^{2}}{4} \right)  + \mathcal{O}(E^{2}).
\]
This explicitly isolates the effect introduced by the proper-time extremization.

A particularly interesting feature of Eq.~\eqref{eq:weak-perlick-trajectory}
is the role played by the terminal time $\tau_{b}$. One might expect that
moving the terminal point farther into the future would suppress its influence
on the motion. At fixed Weyl field, however, precisely the opposite occurs.
For $\tau\ll\tau_{b}$, the dominant contribution is
\begin{equation}
x_{\mathrm{P}}(\tau) \simeq-\frac{E\tau_{b}}{4}\tau^{2},
\label{eq:large-taub-fixed-E}%
\end{equation}
so that the displacement grows linearly with $\tau_{b}$. Therefore, the
relevant parameter controlling the nonlocal correction is not $E$ alone, but
the combination $E\tau_{b}$.

Figure~\ref{fig:large-taub-comparison}(a) illustrates this behavior using the
complete nonlinear equations. The Weyl field is kept fixed while the terminal
time is increased. The Weyl autoparallel is insensitive to $\tau_{b}$, whereas
the Perlick trajectory departs progressively farther from both the Minkowski
and Weyl-autoparallel solutions. Thus, simply taking $\tau_{b}\rightarrow
\infty$ at fixed $E$ does not produce a local limit. Instead, it enhances the
terminal contribution and eventually invalidates a perturbative expansion
based solely on the assumption $E\ll1$.

\begin{figure}[t]
\centering
\begin{minipage}{0.48\textwidth}
\centering
\includegraphics[width=\linewidth]{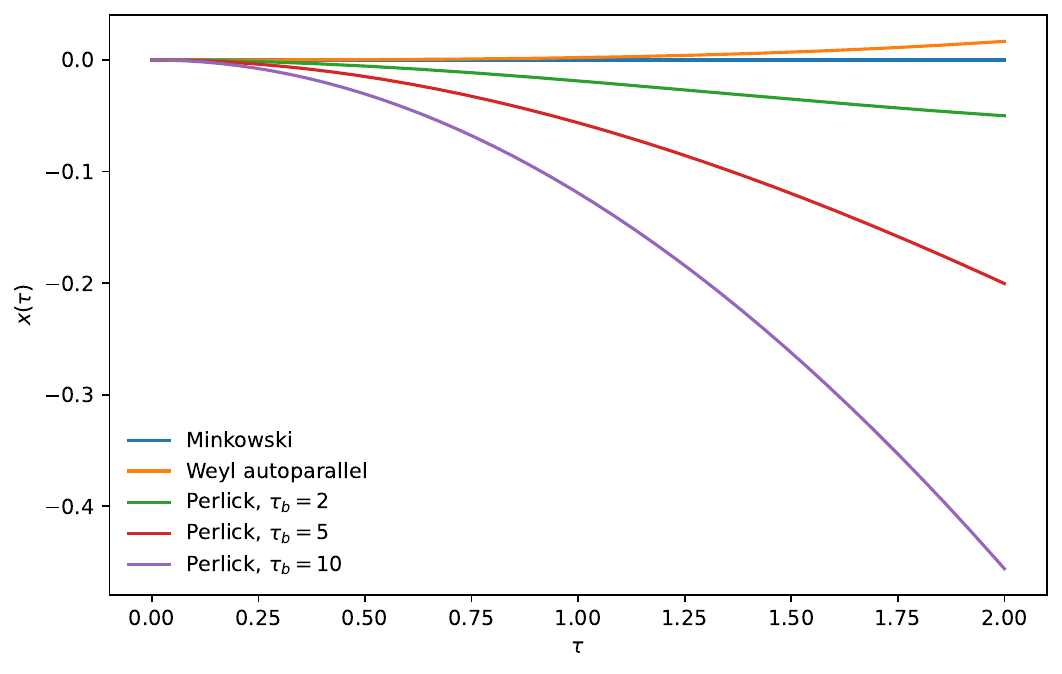}
\textbf{(a)}
\end{minipage}
\hfill\begin{minipage}{0.48\textwidth}
\centering
\includegraphics[width=\linewidth]{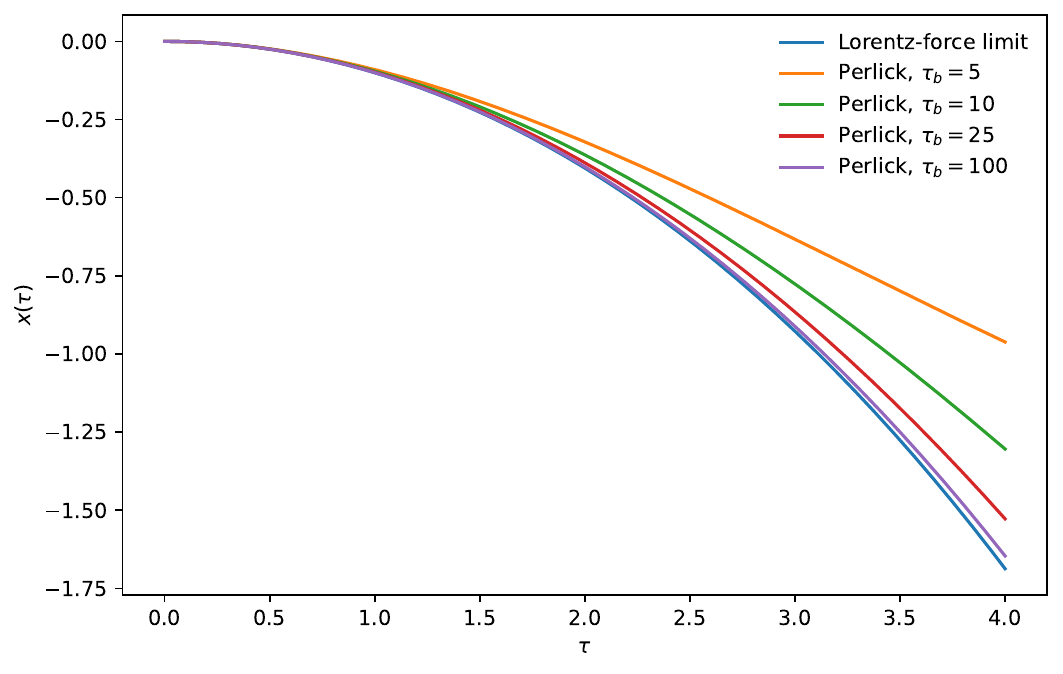}
\textbf{(b)}
\end{minipage}
\caption{ Effect of increasing the terminal proper time $\tau_{b}$ on the
Perlick trajectories. (a) At fixed Weyl field, increasing $\tau_{b}$ enhances
the curvature-dependent contribution and produces a progressively larger
departure from the Minkowski and Weyl-autoparallel trajectories. (b) In the
double-scaling regime, with $\lambda\tau_{b}/2=q/m=0.2$ held fixed, the Weyl
coupling decreases as $\tau_{b}$ increases and the Perlick trajectories
converge towards the ordinary Lorentz-force solution. }%
\label{fig:large-taub-comparison}%
\end{figure}

As discussed in the Lorentz force limit, there is, however, a qualitatively
different large-$\tau_{b}$ regime. If we write the Weyl $1$-form as
$\sigma_{\mu}=\lambda A_{\mu}$, and consider the simultaneous limit
$\lambda\rightarrow0$ and $\tau_{b}\rightarrow\infty$. The equation of motion
consequently approaches the ordinary Lorentz-force equation \eqref{eq:lorentz-force-recovered}.

This behavior is shown in Fig.~\ref{fig:large-taub-comparison}(b). Here the
effective coupling $q/m=\lambda\tau_{b}/2$ is held fixed while $\tau_{b}$ is
increased. The corresponding value of $\lambda$ therefore decreases for each
successive curve. In contrast with Fig.~\ref{fig:large-taub-comparison}(a),
the Perlick trajectories now converge towards a limiting trajectory. This
limiting curve is precisely the one generated by the Lorentz-force equation~\eqref{eq:lorentz-force-recovered}.

The comparison between Figs.~\ref{fig:large-taub-comparison}(a) and
\ref{fig:large-taub-comparison}(b) emphasizes that the large terminal-time
limit is not unique. A distant terminal point does not by itself remove the
nonlocality of the Perlick equation. At fixed Weyl curvature, its influence is
instead amplified. A finite local dynamics emerges only when the weakening of
the Weyl field compensates for the increasing terminal interval. In this
sense, the Lorentz force arises as a double-scaling limit in which the local
Weyl geometry becomes arbitrarily weak while the integrated effect associated
with the remaining worldline stays finite.

This result also provides a useful interpretation of the apparent terminal
character of the original equation. At finite $\tau_{b}$, the motion retains
information about the terminal boundary through the remaining proper-time
interval. In the double-scaling regime, however, this explicit dependence
disappears from observations performed over a finite interval: the terminal
point is pushed arbitrarily far away while its integrated contribution is
encoded in the finite constant $q/m$. The resulting dynamics is local and
indistinguishable, at leading order, from the motion of a charged particle in
an electromagnetic field.

\section{Perlick's geometry and non-locality}

At this point, it is interesting to note that Perlick's proposal leads to a
new kind of differential geometry. Indeed, one can view the equation (1) as a
prescription of how to define length of curves in an entire class of Weyl
manifolds. In other words, Perlick's proper time endows space-time with new
metric properties which are distinct from the pseudo-Riemannian metric of each member
of the class. As a matter of fact, in this paper we tried to answer the
following question: what are the \textquotedblleft geodesics\textquotedblright%
\ of this geometry? This question is both motivated by geometry and physics.
Indeed, one may be interested in knowing whether it is possible or not to
define \textquotedblleft distance\textquotedblright\ between points. Or, one
may regard the geodesics as describing the paths of freely falling particles.
In any case, the equation of the curve which extremizes Perlick's functional
was needed. Since the equation that define Perlick's proper time may be also
viewed as prescription to give a way to define length of curves, then we have
a new (non-local) geometry Taking thiw view the results we have obtained seems
to yield the following:

i) It is possible to obtain the equation of geodesic as the curve that
extremizes the functional defined by (\ref{eq:perlick-proper-time}).

ii) Perlick's geodesics do not coincide with the affine geodesics of Weyl geometry;

iii) It has a non-local character in the sense that it depends on the world
line of the particle. \ 

\section{Final remarks}

\label{sec:final}

In this work, we have investigated the extremization of Perlick's proper time
in a general Weyl geometry. We showed that, contrary to what occurs in the
pseudo-Riemannian case, the extremals of the proper-time functional do not
generally coincide with the autoparallels of the Weyl connection, except in
the integrable case. The difference is governed by the length curvature
$F_{\mu\nu}$ and therefore constitutes a genuine effect of the
non-integrability of the Weyl structure. In the integrable case, where
$F_{\mu\nu}=0$, this additional contribution disappears and the usual
correspondence is recovered.

The equation governing the extremals also exhibits an unusual nonlocal
structure. Through the quantity defined in Eq.~\eqref{eq:J}, the acceleration
at a given point depends on the remaining portion of the worldline up to the
terminal point of the variational interval. In the Weyl proper-time
parametrization, this dependence becomes particularly transparent, being
encoded in the remaining proper time $\tau_{b}-\tau$. The extremal problem
therefore has a terminal character. This does not by itself
imply retrocausal physics, since the variational principle is naturally
formulated as a boundary-value problem, but it distinguishes these extremals
from the trajectories generated by an ordinary local initial-value equation.

The weak constant-field example makes this distinction explicit. Minkowski
geodesics, Weyl autoparallels, and Perlick extremals describe different
trajectories even when they are assigned the same initial position and tangent
vector. Moreover, increasing the terminal time at fixed Weyl field does not
suppress the terminal contribution. On the contrary, as illustrated in
Fig.~\ref{fig:large-taub-comparison}(a), it enhances the departure of the
Perlick extremals from the local trajectories. Thus, a large terminal time
alone does not provide a local limit of the theory.

A different behavior emerges when the large-$\tau_{b}$ limit is accompanied by
a simultaneous weakening of the Weyl field. In the double-scaling regime
discussed in Sec.~\ref{sec:weak-field}, the direct effects of the Weyl
connection vanish while the product of the Weyl coupling and the terminal
proper time remains finite. The resulting dynamics approaches the ordinary
Lorentz force, as illustrated in Fig.~\ref{fig:large-taub-comparison}(b). This
provides an intriguing connection between the length curvature of Weyl
geometry and electromagnetic-like dynamics: a finite local interaction can
survive even when the underlying Weyl field becomes arbitrarily weak, due to
the compensating contribution of an increasingly distant terminal boundary.

The physical significance of this limit remains to be clarified. In
particular, identifying the surviving combination with an intrinsic
charge-to-mass ratio requires a physical prescription for the terminal scale,
which in the original variational problem is simply part of the chosen
interval. Nevertheless, the result shows that familiar local dynamics can
emerge as a limiting regime of the nonlocal equations obtained by extremizing
Perlick's proper time.

Finally, the nonlocal equations admit an equivalent local representation
through the introduction of auxiliary worldline variables, provided the
corresponding terminal condition is retained. Promoting these variables to
independent degrees of freedom and replacing the terminal condition by initial
data would instead define an enlarged local theory. Since such a construction
goes beyond the original Perlick variational problem, its formulation and
physical consequences will be investigated elsewhere.

%%%%%%%%%%%%%%%%%%%%%%%%%%%%%%%%%%%%%%%%%%%%%%%%%%%%%%%%%%%%%%%%%%%%%%%%%%%%%%%%%%%%%%

\section*{Acknowledgements}

\noindent C. Romero would like to CNPq (Brazil) for financial support. I. P.
L. acknowledges partial support from the National Council for Scientific and
Technological Development, CNPq, under grant 312547/2023-4. I. P. L.
acknowledges the networking support by the COST Action BridgeQG (CA23130), the
COST Action RQI (CA23115) and the COST Action FuSe (CA24101) supported by COST
(European Cooperation in Science and Technology).

\ 

\appendix

\section{The case of a delta variation}

\label{app:details}

We consider the functional evaluated along a one-parameter family of curves
$\eta^{\alpha}(s,\varepsilon)$:
\begin{equation}
S[\eta^{\alpha}(s,\varepsilon)] = \int_{a}^{b} \exp\left[  \int_{a}^{s}
l\bigl(\eta^{\alpha}(t,\varepsilon), \dot{\eta}^{\alpha}(t,\varepsilon
)\bigr)\,dt \right]  L\bigl(\eta^{\alpha}(s,\varepsilon), \dot{\eta}^{\alpha
}(s,\varepsilon)\bigr)\,ds.
\end{equation}
To extremize this functional, we must evaluate $\partial S/\partial
\varepsilon$. Differentiating under the integral sign gives
\begin{align}
\frac{\partial S}{\partial\varepsilon} ={}  &  \int_{a}^{b} \exp\left[
\int_{a}^{s} l(\eta,\dot{\eta})\,dt \right]  \Bigg\{ \left[  \frac{\partial
}{\partial\varepsilon} \int_{a}^{s} l(\eta,\dot{\eta})\,dt \right]
L(\eta,\dot{\eta}) + \frac{\partial L}{\partial\varepsilon}
\Bigg\}\,ds\nonumber\\
={}  &  \int_{a}^{b} \exp\left[  \int_{a}^{s} l(\eta,\dot{\eta})\,dt \right]
\Bigg\{ \left[  \int_{a}^{s} \left(  \frac{\partial l}{\partial\eta^{\alpha}}
\frac{\partial\eta^{\alpha}}{\partial\varepsilon} + \frac{\partial l}%
{\partial\dot{\eta}^{\alpha}} \frac{\partial\dot{\eta}^{\alpha}}%
{\partial\varepsilon} \right)  _{\!(t,\varepsilon)} dt \right]
L(s,\varepsilon) + \left[  \frac{\partial L}{\partial\eta^{\alpha}}
\frac{\partial\eta^{\alpha}}{\partial\varepsilon} + \frac{\partial L}%
{\partial\dot{\eta}^{\alpha}} \frac{\partial\dot{\eta}^{\alpha}}%
{\partial\varepsilon} \right]  _{\!(s,\varepsilon)} \Bigg\}\,ds .
\label{eq:variation-first-step}%
\end{align}
Repeated Greek indices are summed, and we use subscripts to emphasize the
functional dependence of the expressions. Using
\begin{equation}
\frac{\partial\dot{\eta}^{\alpha}}{\partial\varepsilon} = \frac{d}{ds} \left(
\frac{\partial\eta^{\alpha}}{\partial\varepsilon} \right)  ,
\end{equation}
we may write
\begin{align}
\frac{\partial l}{\partial\dot{\eta}^{\alpha}} \frac{\partial\dot{\eta
}^{\alpha}}{\partial\varepsilon} ={}  &  \frac{d}{dt} \left(  \frac{\partial
l}{\partial\dot{\eta}^{\alpha}} \frac{\partial\eta^{\alpha}}{\partial
\varepsilon} \right)  - \frac{d}{dt} \left(  \frac{\partial l}{\partial
\dot{\eta}^{\alpha}} \right)  \frac{\partial\eta^{\alpha}}{\partial
\varepsilon},\\
\frac{\partial L}{\partial\dot{\eta}^{\alpha}} \frac{\partial\dot{\eta
}^{\alpha}}{\partial\varepsilon} ={}  &  \frac{d}{ds} \left(  \frac{\partial
L}{\partial\dot{\eta}^{\alpha}} \frac{\partial\eta^{\alpha}}{\partial
\varepsilon} \right)  - \frac{d}{ds} \left(  \frac{\partial L}{\partial
\dot{\eta}^{\alpha}} \right)  \frac{\partial\eta^{\alpha}}{\partial
\varepsilon}.
\end{align}
Substituting these identities into \eqref{eq:variation-first-step}, we obtain
\begin{align}
\frac{\partial S}{\partial\varepsilon} ={}  &  \int_{a}^{b} \exp\left[
\int_{a}^{s} l(\eta,\dot{\eta})\,dt \right]  \Bigg\{ \Bigg[ \int_{a}^{s}
\left(  \frac{\partial l}{\partial\eta^{\alpha}} - \frac{d}{dt} \frac{\partial
l}{\partial\dot{\eta}^{\alpha}} \right)  \frac{\partial\eta^{\alpha}}%
{\partial\varepsilon}\,dt + \frac{\partial l}{\partial\dot{\eta}^{\alpha}}
\frac{\partial\eta^{\alpha}}{\partial\varepsilon} \Bigg]_{\!(s,\varepsilon)}
L(s,\varepsilon)\nonumber\\
&  \hspace{1.5cm} + \left[  \frac{\partial L}{\partial\eta^{\alpha}} -
\frac{d}{ds} \left(  \frac{\partial L}{\partial\dot{\eta}^{\alpha}} \right)  -
l\frac{\partial L}{\partial\dot{\eta}^{\alpha}} \right]  _{\!(s,\varepsilon)}
\frac{\partial\eta^{\alpha}}{\partial\varepsilon} \Bigg\}\,ds \int_{a}^{b}
\frac{d}{ds} \left\{  \exp\left[  \int_{a}^{s} l(\eta,\dot{\eta})\,dt \right]
\frac{\partial L}{\partial\dot{\eta}^{\alpha}} \frac{\partial\eta^{\alpha}%
}{\partial\varepsilon} \right\}  \,ds . \label{eq:general-var-before-boundary}%
\end{align}
The variation vector field along the original curve is
\begin{equation}
\xi^{\alpha}(s) \equiv\left.  \frac{\partial\eta^{\alpha}(s,\varepsilon)}
{\partial\varepsilon} \right|  _{\varepsilon=0}, \qquad\xi^{\alpha}%
(a)=\xi^{\alpha}(b)=0.
\end{equation}
Therefore, the total derivative term in \eqref{eq:general-var-before-boundary}
vanishes. Evaluating the remaining terms at $\varepsilon=0$, where
$\eta^{\alpha}(s,0)=x^{\alpha}(s)$, gives
\begin{align}
\left.  \frac{\partial S}{\partial\varepsilon} \right|  _{\varepsilon=0} ={}
&  \int_{a}^{b} \exp\left[  \int_{a}^{s} l\bigl(x(t),\dot{x}(t)\bigr)\,dt
\right]  \Bigg\{ \left[  \int_{a}^{s} \left(  \frac{\partial l}{\partial
x^{\alpha}} - \frac{d}{dt} \frac{\partial l}{\partial\dot{x}^{\alpha}}
\right)  _{\!t} \xi^{\alpha}(t)\,dt \right]  L_{s}\nonumber\\
&  \hspace{1.5cm} + \left[  \frac{\partial L}{\partial x^{\alpha}} - \frac
{d}{ds} \left(  \frac{\partial L}{\partial\dot{x}^{\alpha}} \right)  +
\frac{\partial l}{\partial\dot{x}^{\alpha}}L - l\frac{\partial L}{\partial
\dot{x}^{\alpha}} \right]  _{\!s} \xi^{\alpha}(s) \Bigg\}\,ds .
\label{eq:first-variation-delta-appendix}%
\end{align}

In Appendix \ref{app:general-var}, we prove the result for an arbitrary
variation. For the moment, we consider a variation localized at $s=\lambda$:
\begin{equation}
\xi^{\alpha}(s) = v^{\alpha}\delta(s-\lambda),
\label{eq:delta-vector-variation}%
\end{equation}
where $v^{\alpha}$ is an arbitrary constant vector. Substitution into
\eqref{eq:first-variation-delta-appendix} yields
\begin{align}
\left.  \frac{\partial S}{\partial\varepsilon} \right|  _{\varepsilon=0} ={}
&  v^{\alpha}\int_{a}^{b} \exp\left[  \int_{a}^{s} l\bigl(x(t),\dot
{x}(t)\bigr)\,dt \right]  \Bigg\{ \left[  \int_{a}^{s} \left(  \frac{\partial
l}{\partial x^{\alpha}} - \frac{d}{dt} \frac{\partial l}{\partial\dot
{x}^{\alpha}} \right)  _{\!t} \delta(t-\lambda)\,dt \right]  L_{s}\nonumber\\
&  \hspace{1.5cm} + \left[  \frac{\partial L}{\partial x^{\alpha}} - \frac
{d}{ds} \left(  \frac{\partial L}{\partial\dot{x}^{\alpha}} \right)  +
\frac{\partial l}{\partial\dot{x}^{\alpha}}L - l\frac{\partial L}{\partial
\dot{x}^{\alpha}} \right]  _{\!s} \delta(s-\lambda) \Bigg\}\,ds .
\label{eq:delta-substitution}%
\end{align}
The first term can be expressed using the Heaviside step function:
\begin{align}
\int_{a}^{s} \left(  \frac{\partial l}{\partial x^{\alpha}} - \frac{d}{dt}
\frac{\partial l}{\partial\dot{x}^{\alpha}} \right)  _{\!t} \delta
(t-\lambda)\,dt  &  = \int_{a}^{b} \theta(s-t) \left(  \frac{\partial
l}{\partial x^{\alpha}} - \frac{d}{dt} \frac{\partial l}{\partial\dot
{x}^{\alpha}} \right)  _{\!t} \delta(t-\lambda)\,dt= \theta(s-\lambda) \left(
\frac{\partial l}{\partial x^{\alpha}} - \frac{d}{d\lambda} \frac{\partial
l}{\partial\dot{x}^{\alpha}} \right)  _{\!\lambda}. \label{eq:heaviside-delta}%
\end{align}
It follows that
\begin{align}
\left.  \frac{\partial S}{\partial\varepsilon} \right|  _{\varepsilon=0} ={}
&  v^{\alpha}\left(  \frac{\partial l}{\partial x^{\alpha}} - \frac
{d}{d\lambda} \frac{\partial l}{\partial\dot{x}^{\alpha}} \right)
_{\!\lambda} \int_{a}^{b} \theta(s-\lambda) \exp\left[  \int_{a}^{s}
l\bigl(x(t),\dot{x}(t)\bigr)\,dt \right]  L_{s}\,ds\nonumber\\
&  + v^{\alpha}\exp\left[  \int_{a}^{\lambda}l\bigl(x(t),\dot{x}(t)\bigr)\,dt
\right]  \Bigg[ \frac{\partial L}{\partial x^{\alpha}} - \frac{d}{d\lambda}
\left(  \frac{\partial L}{\partial\dot{x}^{\alpha}} \right)  + \frac{\partial
l}{\partial\dot{x}^{\alpha}}L - l\frac{\partial L}{\partial\dot{x}^{\alpha}}
\Bigg]_{\!\lambda}. \label{eq:delta-var-result}%
\end{align}
Since
\begin{equation}
\theta(s-\lambda)=
\begin{cases}
0, & s<\lambda,\\
1, & s>\lambda,
\end{cases}
\end{equation}
the integral containing the Heaviside function becomes
\begin{equation}
\int_{a}^{b} \theta(s-\lambda) \exp\left[  \int_{a}^{s} l\bigl(x(t),\dot
{x}(t)\bigr)\,dt \right]  L_{s}\,ds = \int_{\lambda}^{b} \exp\left[  \int
_{a}^{s} l\bigl(x(t),\dot{x}(t)\bigr)\,dt \right]  L_{s}\,ds.
\end{equation}
Imposing the extremality condition
\begin{equation}
\left.  \frac{\partial S}{\partial\varepsilon} \right|  _{\varepsilon=0} =0,
\end{equation}
and using the arbitrariness of $v^{\alpha}$, we find, for each component
$\alpha$,
\begin{align}
\Bigg[ \frac{\partial L}{\partial x^{\alpha}} - \frac{d}{d\lambda} \left(
\frac{\partial L}{\partial\dot{x}^{\alpha}} \right)  + \frac{\partial
l}{\partial\dot{x}^{\alpha}}L - l\frac{\partial L}{\partial\dot{x}^{\alpha}}
\Bigg]_{\!\lambda} ={}  &  \left[  \frac{d}{d\lambda} \left(  \frac{\partial
l}{\partial\dot{x}^{\alpha}} \right)  - \frac{\partial l}{\partial x^{\alpha}}
\right]  _{\!\lambda} \exp\left[  -\int_{a}^{\lambda}l\bigl(x(t),\dot
{x}(t)\bigr)\,dt \right] \nonumber\\
&  \times\int_{\lambda}^{b} \exp\left[  \int_{a}^{s} l\bigl(x(t),\dot
{x}(t)\bigr)\,dt \right]  L\,ds . \label{eq:eom-lambda}%
\end{align}

The last factor of the above equation can be written as $\int_{\lambda}^{b}
\exp\left[  \int_{\lambda}^{s} l\bigl(x(t),\dot{x}(t)\bigr)\,dt \right]
L_{s}\,ds$, which gives equation \eqref{eq:main}. So far, we have shown that
\eqref{eq:eom-lambda} is a necessary condition for extremizing the Perlick
functional. We now verify that it is also sufficient.

For an arbitrary variation vector field $\xi^{\alpha}(s)$, the first variation
is
\begin{align}
\left.  \frac{\partial S}{\partial\varepsilon} \right|  _{\varepsilon=0} ={}
&  \int_{a}^{b} \exp\left[  \int_{a}^{s} l\bigl(x(t),\dot{x}(t)\bigr)\,dt
\right]  \Bigg\{ \left[  \int_{a}^{s} \left(  \frac{\partial l}{\partial
x^{\alpha}} - \frac{d}{dt} \frac{\partial l}{\partial\dot{x}^{\alpha}}
\right)  _{\!t} \xi^{\alpha}(t)\,dt \right]  L_{s}\nonumber\\
&  \hspace{1.5cm} + \left[  \frac{\partial L}{\partial x^{\alpha}} - \frac
{d}{ds} \left(  \frac{\partial L}{\partial\dot{x}^{\alpha}} \right)  +
\frac{\partial l}{\partial\dot{x}^{\alpha}}L - l\frac{\partial L}{\partial
\dot{x}^{\alpha}} \right]  _{\!s} \xi^{\alpha}(s) \Bigg\}\,ds .
\label{eq:arbitrary-first-variation}%
\end{align}
Substituting the equations of motion \eqref{eq:eom-lambda}, we obtain
\begin{align}
\left.  \frac{\partial S}{\partial\varepsilon} \right|  _{\varepsilon=0} ={}
&  \int_{a}^{b} \exp\left[  \int_{a}^{s} l\bigl(x(t),\dot{x}(t)\bigr)\,dt
\right]  L_{s} \left[  \int_{a}^{s} \left(  \frac{\partial l}{\partial
x^{\alpha}} - \frac{d}{dt} \frac{\partial l}{\partial\dot{x}^{\alpha}}
\right)  _{\!t} \xi^{\alpha}(t)\,dt \right]  ds\nonumber\\
&  - \int_{a}^{b} \left[  \left(  \frac{\partial l}{\partial x^{\alpha}} -
\frac{d}{ds} \frac{\partial l}{\partial\dot{x}^{\alpha}} \right)  _{\!s}
\xi^{\alpha}(s) \right]  \int_{s}^{b} \exp\left[  \int_{a}^{u}
l\bigl(x(t),\dot{x}(t)\bigr)\,dt \right]  L_{u}\,du\,ds .
\label{eq:sufficiency-expanded}%
\end{align}
Define the functions
\begin{align}
F(s)  &  = \int_{a}^{s} \left(  \frac{\partial l}{\partial x^{\alpha}} -
\frac{d}{dt} \frac{\partial l}{\partial\dot{x}^{\alpha}} \right)  _{\!t}
\xi^{\alpha}(t)\,dt, & \frac{dF}{ds}  &  = \left(  \frac{\partial l}{\partial
x^{\alpha}} - \frac{d}{ds} \frac{\partial l}{\partial\dot{x}^{\alpha}}
\right)  _{\!s} \xi^{\alpha}(s),\\
G(s)  &  = -\int_{s}^{b} \exp\left[  \int_{a}^{u} l\bigl(x(t),\dot
{x}(t)\bigr)\,dt \right]  L_{u}\,du, & \frac{dG}{ds}  &  = \exp\left[
\int_{a}^{s} l\bigl(x(t),\dot{x}(t)\bigr)\,dt \right]  L_{s} .
\end{align}
Equation \eqref{eq:sufficiency-expanded} therefore becomes
\begin{align}
\left.  \frac{\partial S}{\partial\varepsilon} \right|  _{\varepsilon=0}=
\int_{a}^{b} \left[  \frac{dG}{ds}F(s) + G(s)\frac{dF}{ds} \right]  ds=
\int_{a}^{b} \frac{d}{ds}\bigl[F(s)G(s)\bigr]\,ds = F(b)G(b)-F(a)G(a) =0,
\end{align}
because $F(a)=0$ and $G(b)=0$, independently of the choice of $\xi^{\alpha
}(s)$.

Thus, satisfying the equations of motion \eqref{eq:eom-lambda} is both a
necessary and a sufficient condition for a trajectory to extremize the Perlick
functional \eqref{eq:new_perlick}.

\section{Extremization for a general variation}

\label{app:general-var}

\begin{proposition}
Let $\mathcal{U}\subset\mathbb{R}^{n}$ be an open set; $l:T\mathcal{U}%
\rightarrow\mathbb{R}$ and $L:\mathbb{R} \times T\mathcal{U}\rightarrow
\mathbb{R}$ be smooth functions; and $x_{0}$, $x_{1} \in\mathcal{U}$, $t_{0} <
t_{1}$ real numbers. Let $S$ be the functional defined on the set of all
smooth curves $\gamma:[t_{0},t_{1}]\rightarrow\mathcal{U}$ such that
$\gamma(t_{0})=x_{0}$ and $\gamma(t_{1})=x_{1}$, that is given by
\begin{equation}
S[\gamma]= \int_{t_{0}}^{t_{1}}L\Big( \int_{t_{0}}^{t}l(\gamma(s),\dot{\gamma
}(s)) \, ds , \gamma(t), \dot{\gamma}(t) \Big) \, dt,
\label{eq:general-functional}%
\end{equation}
where $(\gamma(t),\dot{\gamma}(t)) \in\mathcal{U}\times\mathbb{R}^{n} \cong
T\mathcal{U}$.

Then $\gamma$ is a critical point of $S$ if and only if it satisfies the
generalized Euler Lagrange equation:
\begin{equation}
\frac{\partial L}{\partial x^{a}} - \frac{d}{dt} \frac{\partial L}%
{\partial\dot{x}^{a}} +\frac{\partial L}{\partial l} \frac{\partial
l}{\partial\dot{x}^{a}} +\Big[ \frac{\partial l}{\partial x^{a}} - \frac
{d}{dt}\frac{\partial l}{\partial\dot{x}^{a}} \Big] \int_{t}^{t_{1}}
\frac{\partial L}{\partial l} ds =0, \label{eq:generalized-euler-lagrange}%
\end{equation}
for each $a\in\{1,2,...,n\}$, where $\frac{\partial L}{\partial l}$ denotes
the derivative of $L$ with respect to its first coordinate (where we put the
integral of $l$). In this equation, all the derivatives of $L$ are being
evaluated at $\Big( \int_{t_{0}}^{t}l(\gamma(r),\dot{\gamma}(r)) \, dr ,
\gamma(t), \dot{\gamma}(t) \Big)$, with the exception of the last term which
is being integrated, which is being evaluated at $\Big( \int_{t_{0}}%
^{s}l(\gamma(r),\dot{\gamma}(r)) \, dr , \gamma(s), \dot{\gamma}(s) \Big)$;
and all derivatives of $l$ are being evaluated at $(\gamma(t), \dot{\gamma
}(t))$.
\end{proposition}

\begin{proof}
Let $\eta:[t_0,t_1] \rightarrow \mathbb{R}^n$ be a smooth curve such that $\eta(t_0)=\eta(t_1)=0$. For small enough $\epsilon \in \mathbb{R}$ the image of the function $\gamma + \epsilon \eta$ lies in $\mathcal{U}$. We need to show that for an arbitrary $\eta$
\begin{equation}
\frac{\partial S[\gamma+\epsilon\eta]}{\partial \epsilon}\Big|_{\epsilon = 0} = 0
\label{eq:critical-point-condition}
\end{equation}
if and only if \eqref{eq:generalized-euler-lagrange} holds.
First, let us assume \eqref{eq:critical-point-condition} and prove \eqref{eq:generalized-euler-lagrange}. Using the chain rule and the rules of differentiating functions defined by integrals we can rewrite \eqref{eq:critical-point-condition}, using the definition \eqref{eq:general-functional}, as
$$
\frac{\partial}{\partial \epsilon}
\int_{t_0}^{t_1}L\Big( \int_{t_0}^tl(\gamma(s)+\epsilon \eta(s),\dot{\gamma}(s)+\epsilon \eta'(s)) \, ds , \gamma (t)+\epsilon \eta(t), \dot{\gamma}(t)+\epsilon \eta'(t) \Big) \, dt
\Big|_{\epsilon=0} =
$$
\begin{equation}
\int_{t_0}^{t_1}
\Big\{
\frac{\partial L}{\partial l}
\Big[ \int_{t_0}^t\frac{\partial l}{\partial x^a}\eta^a + \frac{\partial l}{\partial \dot{x}^a} \eta'^a \, ds \Big]
+\frac{\partial L}{\partial x^a}\eta^a+ \frac{\partial L}{\partial \dot{x}^a}\eta'^a
\Big\} \, dt.
\label{eq:first-variation-general}
\end{equation}
Now, we use integral by parts to get rid of all $\eta'^a$. With respect to the $\eta'^a$ outside the inner integral, everything is just as usual, since the outer integral has extremes $t_0$ and $t_1$, where $\eta$ is zero. With respect to the $\eta'^a$ inside the inner integral, this process will give us an extra term, since $\eta$ doesn't have to be zero at $t$. So, \eqref{eq:first-variation-general} can be rewritten as
\begin{equation}
\int_{t_0}^{t_1}
\Big\{
\frac{\partial L}{\partial l}
\Big[ \int_{t_0}^t\frac{\partial l}{\partial x^a}\eta^a -\frac{d}{ds}\Big( \frac{\partial l}{\partial \dot{x}^a}\Big) \eta^a \, ds \Big]
+\frac{\partial L}{\partial l} \frac{\partial l}{\partial \dot{x}^a} \eta ^a
+\frac{\partial L}{\partial x^a}\eta^a- \frac{d}{dt} \Big(\frac{\partial L}{\partial \dot{x}^a} \Big)\eta^a
\Big\} \, dt.
\label{eq:first-variation-after-parts}
\end{equation}
At this point, we observe that if we set $\eta^a$ to be equal to a Dirac delta and $\eta^b$, for $b \neq a$ to be zero, we would obtain \eqref{eq:generalized-euler-lagrange} directly. We proceed analytically.
Define the function $h:[t_0,t_1) \rightarrow \mathbb{R}$ by
$$
h(t^*)=
$$
\begin{equation}
\int_{t^*}^{t_1}
\Big\{
\frac{\partial L}{\partial l}
\Big[ \int_{t^*}^t\frac{\partial l}{\partial x^a}\eta^a -\frac{d}{ds}\Big( \frac{\partial l}{\partial \dot{x}^a}\Big) \eta^a \, ds \Big]
+\frac{\partial L}{\partial l} \frac{\partial l}{\partial \dot{x}^a} \eta ^a
+\frac{\partial L}{\partial x^a}\eta^a- \frac{d}{dt} \Big(\frac{\partial L}{\partial \dot{x}^a} \Big)\eta^a
\Big\} \, dt.
\label{eq:def-h-general}
\end{equation}
Which is just the variation at \eqref{eq:first-variation-after-parts} with $t^*$ instead of $t_0$. We assert that $h(t^*)\equiv 0$. Indeed, let $t^* \in [t_0,t_1)$, since we are assuming \eqref{eq:critical-point-condition} for arbitrary $\eta$, we might as well take, for any natural number $k$, $\eta_k$ instead of $\eta$, where $\eta_k = f_k \eta$, and $f_k:[t_0,t_1]\rightarrow\mathbb{R}$ is a smooth function such that there is a $r_k \in (t^*,t_1)$, $r_k-t^*<1/k$, so that $f$ is given by (figure 1):
\begin{enumerate}
\item If $t \in [t_0,t^*]$ then $f(t)=0$;
\item if $t \in (t^*,r_k)$ then $0 \leq f(t) \leq 1$;
\item if $t \in [r_k,t_1]$ then $f(t) = 1$.
\end{enumerate}
By \eqref{eq:critical-point-condition}, \eqref{eq:first-variation-after-parts}, and definition of $\eta_k$,
$$
0=
$$
$$
\int_{t_0}^{t_1}
\Big\{
\frac{\partial L}{\partial l}
\Big[ \int_{t_0}^t\frac{\partial l}{\partial x^a}\eta^a_k -\frac{d}{ds}\Big( \frac{\partial l}{\partial \dot{x}^a}\Big) \eta^a_k \, ds \Big]
+\frac{\partial L}{\partial l} \frac{\partial l}{\partial \dot{x}^a} \eta ^a_k
+\frac{\partial L}{\partial x^a}\eta^a_k- \frac{d}{dt} \Big(\frac{\partial L}{\partial \dot{x}^a} \Big)\eta^a_k
\Big\} \, dt=
$$
\begin{equation}
\int_{t^*}^{t_1}
\Big\{
\frac{\partial L}{\partial l}
\Big[ \int_{t^*}^t\frac{\partial l}{\partial x^a}\eta^a_k -\frac{d}{ds}\Big( \frac{\partial l}{\partial \dot{x}^a}\Big) \eta^a_k \, ds \Big]
+\frac{\partial L}{\partial l} \frac{\partial l}{\partial \dot{x}^a} \eta ^a_k
+\frac{\partial L}{\partial x^a}\eta^a_k- \frac{d}{dt} \Big(\frac{\partial L}{\partial \dot{x}^a} \Big)\eta^a_k
\Big\} \, dt.
\label{eq:localized-variation-general}
\end{equation}
Taking the limit as $k \rightarrow \infty$, $\eta_k \rightarrow \eta$ on $(t^*,t_1]$. It is an exercise on topology of metric spaces to show that the limit of \eqref{eq:localized-variation-general} as $k$ goes to infinity is precisely $h(t^*)$. Taking the derivative of $h$ on one hand we have zero, on the other we have the derivative of \eqref{eq:def-h-general} with respect to $t^*$. We arrive at
\begin{equation}
h'(t)=
\Big\{
\frac{\partial L}{\partial x^a} - \frac{d}{dt} \frac{\partial L}{\partial \dot{x}^a}
+\frac{\partial L}{\partial l} \frac{\partial l}{\partial \dot{x}^a}
+\Big[
\frac{\partial l}{\partial x^a} - \frac{d}{dt}\frac{\partial l}{\partial \dot{x}^a}
\Big]
\int_t^{t_1} \frac{\partial L}{\partial l} ds \Big\} \eta^a(t) =0.
\label{eq:h-derivative-general}
\end{equation}
By the arbitrariness of $\eta^a$ we conclude \eqref{eq:generalized-euler-lagrange}.
Conversely, assume that $\gamma$ satisfies \eqref{eq:generalized-euler-lagrange}, then by the first equality in \eqref{eq:h-derivative-general}, $h'(t)$ has to be zero. So $h$ is constant. But by \eqref{eq:def-h-general} we have
$$
\lim_{t \to t_1} h(t) = 0,
$$
hence, $h\equiv 0$, which implies that $h(t_0)=0$. By \eqref{eq:def-h-general}, \eqref{eq:first-variation-after-parts} and \eqref{eq:critical-point-condition}, $\gamma$ is a critical point of $S$, as we wanted to show.
\end{proof}

\bibliographystyle{utphys}
\bibliography{bib-weyl}

\end{document}